\documentclass[reprint,amsmath,amssymb,aps,pra,showpacs]{revtex4-1}
\usepackage{amsthm,mathtools}
\usepackage{bm,amsfonts}
\usepackage{graphicx,graphics}
\usepackage{dcolumn}
\usepackage{hyperref}
\usepackage{comment}
\usepackage{color}
\usepackage{epstopdf}

\newcommand{\bra}[1]{\langle#1|}
\newcommand{\ket}[1]{|#1\rangle}

\newcommand{\beq}{\begin{equation}}
\newcommand{\eeq}{\end{equation}}

\newtheorem{theorem}{Theorem}
\newtheorem{lemma}{Lemma}
\newtheorem{exmp}{Example}

\newcommand{\tens}[2]{\mathbin{\mathop{\otimes}\limits_{#1}^{#2}}}
\DeclareMathOperator*{\argmax}{arg\,max}

\begin{document}

\title{Method for quantum-jump continuous-time quantum error correction}
\author{Kung-Chuan Hsu}
  \email{kungchuh@usc.edu}
\author{Todd A. Brun}
  \email{tbrun@usc.edu}
\affiliation{Department of Electrical Engineering, University of Southern California, Los Angeles, California 90089, USA}

\begin{abstract}
Continuous-time quantum error correction (CTQEC) is a technique for protecting quantum information against decoherence, where both the decoherence and error correction processes are considered continuous in time. Given any [[$n,k,d$]] quantum stabilizer code, we formulate a class of protocols to implement CTQEC, involving weak coherent measurements and weak unitary corrections. Under this formalism, we show that the minimal required size of the ancillary system is $n-k+1$ qubits, and we propose one scheme that meets this minimal requirement. Furthermore, we compare our method with other known schemes, and show that a particular measure of performance described in this paper is better when using our method.
\end{abstract}


\maketitle


\section{Introduction}
\label{I}

Continuous-time quantum error correction (CTQEC) studies quantum error correction when both decoherence and error correction are treated as continuous in time. One of the first approaches to CTQEC was proposed by Paz and Zurek~\cite{Paz98}, where the process of error correction is modeled as a weak quantum-jump operation taking a small time $\delta t$. In the limit of $\delta t$ going to zero, the evolution of the code system can be described by a continuous-time master equation. In~\cite{OreshkovChapter}, Oreshkov proposed a quantum-jump CTQEC protocol and reviewed various topics on CTQEC not restricted to the quantum-jump model. Another scheme for CTQEC was proposed by Ahn, Doherty, and Landahl (ADL)~\cite{Ahn02} involving the method of feedback control by estimation~\cite{Doherty99}. Several studies based on this scheme have been published since then (see, e.g.,~\cite{Ahn03,Chase08}). Other work or resource on CTQEC and continuous-time quantum feedback control includes~\cite{Wiseman94,Doherty00,Brun02,Ahn04,Jacobs04,Sarovar04,Sarovar05,Jacobs06,Oreshkov07,Mascarenhas10,Ippoliti15,WisemanMilburn,Jacobs}.

In Sec.~\ref{II}, we introduce the technical background needed for the rest of this paper. In Sec.~\ref{III}, we discuss a class of protocols that perform quantum-jump CTQEC. Furthermore, we show the minimal requirement on the size of the ancillary system and propose a protocol that meets this requirement. In Sec.~\ref{IV}, we compare this protocol to those of Oreshkov and ADL. In addition, we investigate whether optimizing the parameter in the error-correction process of Oreshkov yields any significant improvement in performance, and we show that the improvement is negligibly small in general.


\section{Preliminaries}
\label{II}

In this section, we review background knowledge and notation for quantum stabilizer codes and quantum-jump CTQEC sufficient for the reader to understand the content of this paper.

\subsection{Structure of stabilizer codes}
\label{IIA}

Let $C$ be a stabilizer code that protects $k$ qubits by encoding them into $n$ qubits. Associated with $C$ is a set of $n-k$ stabilizer generators $G=\{g_1,g_2,...,g_{n-k}\}\subset\mathcal{G}^n$, where $\mathcal{G}^n=\{i^j \tens{l=1}{n} O_l: j\in\{0,1,2,3\}, O_l\in\{I,X,Y,Z\},\forall l\}$ is the $n$-fold Pauli group, and $X,Y,Z$ are the Pauli operators. The stabilizer group $S$ is an Abelian subgroup of $\mathcal{G}^n$ that is generated by the stabilizer generators and does not contain $-I^{\otimes n}$. The codespace is defined to be the space $C_S=\left\{\ket{\psi}\in\mathcal{H}^n : g\ket{\psi}=\ket{\psi},~\forall~g\in S\right\}$, where $\mathcal{H}^n$ is the $2^n$-dimensional Hilbert space of $n$ qubits. The code can correct any set of correctable errors $E=\{E_j\}_{j\in J}$ where $E_j E_l \in S\cup\left(\mathcal{G}^n-N(S)\right)~\forall~j,l\in J$, where $J$ is an index set and $N(S)$ is the normalizer of $S$.

The above outlines the structure of stabilizer codes. For more detail, see~\cite{Gottesman97}. In the following, we will introduce three bases of $\mathcal{H}^n$ that are useful in analyzing stabilizer codes:

1. ``Physical basis": The physical basis is the simplest basis, in the sense that physical errors act on the state that is represented in this basis. In this basis, the codespace is $C_S$ as defined above.

2. ``Encoded basis": The state represented in the encoded basis has a set of $k$ information qubits that are to be protected, and the remaining $n-k$ qubits are the syndrome qubits that represent the redundancies used in the error-correcting code. Without loss of generality, we will assign the first $k$ qubits in this basis as the information subsystem and the last $n-k$  qubits as the syndrome subsystem. In this basis, a state is in the codespace if and only if the syndrome qubits are in the state $\ket{0}^{\otimes n-k}$. There exists an encoding unitary $U_\mathcal{E}$ that takes the operators and states represented in the encoded basis to the physical basis.

3. ``Corrected basis": The state in the corrected basis results from applying a correcting unitary to the state in the encoded basis such that any correctable error acting on a codeword state leaves the information subsystem invariant. The correcting unitary, $U_\Gamma$ leaves the syndrome subsystem unchanged, and it applies a unitary to the information system conditioned on which syndrome subspace the state is in. It can be written in the following form in terms of the set of standard basis states of the syndrome subsystem, $\left\{\ket{\underline{j}}\right\}$, and some set of unitary operators, $\left\{U_{\Gamma,j}\right\}$, where $j\in\{0,1,...,2^{n-k}-1\}$:
\beq\label{eq:IIA01}
U_\Gamma=\sum\limits_{j=0}^{2^{n-k}-1} U_{\Gamma,j}\otimes \ket{\underline{j}}\bra{\underline{j}}.
\eeq

Specifically, each $\ket{\underline{j}}$ is defined as a tensor product of standard basis states $\ket{0}$ and $\ket{1}$ in the following way:
\begin{align}
\ket{\underline{j}} &=\ket{j_{n-k}}\otimes\ket{j_{n-k-1}}\otimes...~\ket{j_1}\label{eq:IIA02} \\
&=\ket{j_{n-k}~j_{n-k-1}~...~j_1},\label{eq:IIA03}
\end{align}
where $j_l\in\{0,1\}$ satisfies
\beq\label{eq:IIA04}
j=\sum\limits_{l=1}^{n-k} j_l 2^{l-1}.
\eeq

In other words, $j_{n-k}~j_{n-k-1}~...~j_1$ is the binary representation of $j$ with the most significant digit on the left, and an underline is used to differentiate the index as a decimal number from its binary representation.

The following illustrates the contents in this section using the three-qubit bit-flip code.

\begin{exmp}\label{exmp1}
The three-qubit bit-flip code encodes $k=1$ qubit into $n=3$ qubits. The $n-k=2$ stabilizer generators may be chosen as:
\beq\label{eq:IIA05}
g_1=Z\otimes Z\otimes I,~g_2=Z\otimes I\otimes Z.
\eeq

A correctable set of errors for this code contains the trivial error, $E_{0}=I^{\otimes 3}$, and the bit-flip errors on each of the $n=3$ qubits:
\beq\label{eq:IIA06}
E_{1}=X\otimes I\otimes I,~E_{2}=I\otimes X\otimes I,~E_{3}=I\otimes I\otimes X.
\eeq

An encoding unitary for this code is:
\beq\label{eq:IIA07}
U_\mathcal{E}=\ket{0}\bra{0}\otimes I\otimes I+\ket{1}\bra{1}\otimes X\otimes X.
\eeq

Now, by applying the inverse encoding unitary, $U_\mathcal{E}^{-1}$, we can represent the operators in the encoded basis. Adding the subscript ``(E)" to differentiate them from their original form in the physical basis, they are:

\beq\label{eq:IIA08}
g_{1(E)}=I\otimes Z\otimes I,~g_{2(E)}=I\otimes I\otimes Z.
\eeq
and
\begin{align}
E_{0(E)}&=I^{\otimes 3}~,~~ E_{1(E)}=X\otimes X\otimes X,\label{eq:IIA09} \\
E_{2(E)}&=I\otimes X\otimes I~~,~~ E_{3(E)}=I\otimes I\otimes X.\label{eq:IIA10}
\end{align}

In the encoded basis, just like the physical basis, the correctable error $E_{0(E)}$ leaves any state invariant. So $U_{\Gamma,0}=I$.

On the other hand, the correctable error $E_{1(E)}$ flips every qubit in the system. When the state was previously in a codeword state, $E_{1(E)}$ brings the syndrome qubits from the $\ket{\underline{0}}$-subspace to the $\ket{\underline{3}}$-subspace. Since the information qubit has been flipped as well, the corresponding unitary correction is to apply a bit-flip operator so that the information qubit is flipped back, that is, $U_{\Gamma,3}=X$.

Similarly, one may find that $U_{\Gamma,1}=I,~U_{\Gamma,2}=I$.

Therefore, the correcting unitary defined earlier in this example is:
\begin{align}
U_\Gamma&=\sum\limits_{j=0}^3 U_{\Gamma,j}\otimes\ket{\underline{j}}\bra{\underline{j}}\label{eq:IIA11} \\
   &=I\otimes\left( \ket{\underline{0}}\bra{\underline{0}}+\ket{\underline{1}}\bra{\underline{1}}+\ket{\underline{2}}\bra{\underline{2}} \right)+X\otimes\ket{\underline{3}}\bra{\underline{3}}.\label{eq:IIA12}
\end{align}

To represent the operators in the corrected basis, one applies the unitary $U_\Gamma U_\mathcal{E}^{-1}$ to the operators in the physical basis. We denote the operators in the corrected basis by adding the subscript ``(C)," and they are:
\beq\label{eq:IIA13}
g_{1(C)}=I\otimes Z\otimes I,~g_{2(C)}=I\otimes I\otimes Z,
\eeq
and
\begin{align}
E_{0(C)}&=I^{\otimes 3},\label{eq:IIA14} \\
E_{1(C)}&=I\otimes\left(\ket{\underline{0}}\bra{\underline{3}}+\ket{\underline{3}}\bra{\underline{0}}\right)+X\otimes\left(\ket{\underline{1}}\bra{\underline{2}}+\ket{\underline{2}}\bra{\underline{1}}\right),\label{eq:IIA15} \\
E_{2(C)}&=I\otimes\left(\ket{\underline{0}}\bra{\underline{2}}+\ket{\underline{2}}\bra{\underline{0}}\right)+X\otimes\left(\ket{\underline{1}}\bra{\underline{3}}+\ket{\underline{3}}\bra{\underline{1}}\right),\label{eq:IIA16} \\
E_{3(C)}&=I\otimes\left(\ket{\underline{0}}\bra{\underline{1}}+\ket{\underline{1}}\bra{\underline{0}}\right)+X\otimes\left(\ket{\underline{2}}\bra{\underline{3}}+\ket{\underline{3}}\bra{\underline{2}}\right).\label{eq:IIA17}
\end{align}

Any correctable error leaves the information qubit of a codeword state invariant. The only effect of a correctable error acting on a codeword state is to bring the syndrome system from the $\ket{\underline{0}}$-subspace to a nontrivial syndrome subspace.
\end{exmp}

As will be seen, we will mainly work in the corrected basis since it simplifies the analysis of a CTQEC protocol. However, some caution is required. As written, the corrected basis looks like a tensor product of bases on each of the qubits, but it is not. In the example, the states $\ket{0}$ and $\ket{1}$ for the system qubit are different in syndrome $\ket{\underline{3}}$ than in syndromes $\ket{\underline{0}}$, $\ket{\underline{1}}$ and $\ket{\underline{2}}$. We will therefore also use the operators represented in the physical basis for practical purposes.

\subsection{The model of quantum-jump CTQEC}
\label{IIB}

In conventional quantum error correction, the system undergoes an error-correcting operation, which is characterized by a CPTP mapping $\mathcal{R}\left(\cdot\right)$, consisting of a syndrome measurement followed by a unitary correction operation based on the measurement result.

In quantum-jump CTQEC, a weak version of the error-correcting operation is applied repeatedly on the system. Let $\delta t$ be the small time step of a single instance of the weak error-correcting operation. The infinitesimal error correction map is:
\beq\label{eq:IIB01}
\rho\xrightarrow{\delta t} \left(1-\kappa\delta t\right)\rho+\kappa\delta t\mathcal{R}(\rho),
\eeq
where $\kappa$ is the rate of error correction.

In the limit of $\delta t\rightarrow 0$, the above infinitesimal error-correcting operation approaches a master equation describing the system evolution under error correction:
\beq\label{eq:IIB02}
\frac{d\rho}{dt}=\kappa\left(\mathcal{R}(\rho)-\rho\right).
\eeq

In addition, the system undergoes evolution caused by decoherence. We model the decoherence as Markovian with the system Hamiltonian $H$ and a set of Lindblad operators $\{L_j\}$. Then the full master equation becomes
\beq
\frac{d\rho}{dt}=\mathcal{L}(\rho)+\kappa\left(\mathcal{R}(\rho)-\rho\right),
\eeq
where
\beq
\mathcal{L}(\rho)=-i\left[H,\rho\right]+\sum\limits_{j}\left(L_j\rho L_j^\dag-\frac{1}{2}L_j^\dag L_j\rho-\frac{1}{2}\rho L_j^\dag L_j\right).
\eeq

In this study, we will use quantum error-correcting codes and error models based on standard quantum error correction. Our main goal in exploring quantum-jump CTQEC is to find a method to realize the weak mapping given by Eq.~(\ref{eq:IIB01}).

A general way of realizing this weak mapping involves the following steps:

(1) Couple the primary system to an ancillary system.

(2) Apply a weak unitary to both systems.

(3) Weakly measure and then discard the ancillary system.

(4) Apply a weak unitary correction, conditioned on the measurement result, to the system.

In this paper, we will propose a specific realization of CTQEC based on the above scheme, which will require analysis on the ancillary system used in the coherent measurement steps. Throughout this paper, we will denote a basis state of the ancillary system as $\ket{\underline{j}_a}$, where the subscript ``$a$" signifies that the state is of the ancillary system, and the index $j$ runs from $0$ to the dimension of the ancillary system minus one.


\section{Quantum-jump CTQEC with minimal ancillas}
\label{III}

In~\cite{OreshkovChapter}, a quantum-jump CTQEC protocol was proposed that requires the size of the ancillary system to be the dimension of the syndrome space minus one. That is, for a syndrome system of size $n-k$, an ancillary system of size $2^{n-k}-1$ is required. This exponential overhead in the size of the ancillary system would be a very expensive resource in practice. In this section, we propose a quantum-jump CTQEC protocol that requires only linear overhead, $n-k+1$ to be exact, of ancillary qubits to perform the same task. We will also show that $n-k+1$ is in fact optimal. The main idea is as the following.

Consider CTQEC using an arbitrary [[$n,k,d$]] stabilizer code. In the corrected basis, without loss of generality, we assume the information subsystem is the first $k$ qubits. The last $n-k$ qubits form the syndrome subsystem which absorbs the effect of any correctable error on the code. The syndrome system is initially prepared in the trivial state $\ket{0}^{\otimes n-k}$. When a correctable error happens, some syndrome qubits may be flipped to $\ket{1}$, while the information qubits are left invariant. The goal of error correction is obviously to flip the syndrome qubits back to the trivial state $\ket{0}^{\otimes n-k}$ so that the syndrome subsystem is capable of absorbing any new correctable error that may occur. Such an error-correcting operation may be described in the corrected basis as a CPTP map $\mathcal{R}(.)$:
\beq\label{eq:III01}
\rho\longrightarrow \mathcal{R}(\rho)=\sum\limits_{j=0}^{2^{n-k}-1} R_j\rho R_j^\dag,
\eeq
where
\beq\label{eq:III02}
R_j=I^{\otimes k}\otimes\ket{\underline{0}}\bra{\underline{j}}~,~~\forall~j\in\{0,1,...,2^{n-k}-1\},
\eeq
and $\rho$ is the density operator of the encoded system.

With the given error-correcting map $\mathcal{R}\left(\cdot\right)$, the weak error-correcting map is
\beq\label{eq:III03}
\rho\rightarrow\left(1-\varepsilon^2\right)\rho+\varepsilon^2\mathcal{R}(\rho).
\eeq

Note that, as in Eq.~(\ref{eq:IIB01}), if the map takes time $\delta t$ we define a rate $\kappa$ by
\beq\label{eq:III04}
\varepsilon=\sqrt{\kappa\delta t}.
\eeq

In Sec.~\ref{IIIA}, we first give the structure of a class of methods that implements the CTQEC map in Eq.~(\ref{eq:III03}). In Sec.~\ref{IIIB}, we discuss the requirements on the ancillary system in the coherent measurement steps, and we show that the size of the ancillary system must be at least $n-k+1$. In Sec.~\ref{IIIC}, we propose a CTQEC protocol for any [[$n,k,d$]] stabilizer code, requiring $n-k+1$ ancillary qubits.

\subsection{Preliminary derivation}
\label{IIIA}

In this section, we derive requirements to construct CTQEC protocols based on a given [[$n,k,d$]] stabilizer code. Recall that, in the corrected basis, the CPTP map we wish to achieve is given by
\beq\label{eq:IIIA01}
\rho\rightarrow\left(1-\varepsilon^2\right)\rho+\varepsilon^2\sum\limits_{j=0}^{2^{n-k}-1}R_j\rho R_j^\dag,
\eeq
where
\beq\label{eq:IIIA02}
R_j=I^{\otimes k}\otimes\ket{\underline{0}}\bra{\underline{j}},~\forall j\in\{0,1,...,2^{n-k}-1\}.
\eeq

The above map is already in a Kraus decomposition with Kraus operators
\beq\label{eq:IIIA03}
\left\{\sqrt{1-\varepsilon^2}I^{\otimes n},\varepsilon R_j:j=0,1,...,2^{n-k}-1\right\}.
\eeq

However, this set of Kraus operators is not useful to us, since we are performing continuous-time correction and we need each of the Kraus operators $K$ to be weak, in the sense that it has the form
\beq\label{eq:IIIA04}
K=c\left(I^{\otimes n}+\mathcal{O}(\varepsilon)\right),
\eeq
where $c$ is some constant and $\mathcal{O}(\varepsilon)$ is some term upper-bounded by the order of $\varepsilon$. That is, each of the Kraus operators is close to the identity.

Now, let $\{K_l\}_{l\in S}$ be a set of non-zero weak Kraus operators and $S=\{0,1,...,|S|-1\}$ be an index set that is to be determined. The following lemma (Theorem 8.2 from~\cite{NielsenChuang}) will be useful from here on.
\begin{lemma}\label{lemma1}
Let $M,N\in\{0\}\cup\mathbb{N}$, and let $\{E_j\}_{j=0}^M$ and $\{F_l\}_{l=0}^N$ be sets of Kraus operators for CPTP mapping $\mathcal{E}$ and $\mathcal{F}$, repectively. If $M\neq N$, by appending zero operators to the shorter list of the two sets of Kraus operators we may ensure that $M=N$. Then $\mathcal{E}=\mathcal{F}$ if and only if there exists an $(M+1)\times (M+1)$ unitary matrix $U$, with the $j,l$-th element being $u_{j,l}$, such that $E_j=\sum\limits_{l=0}^M u_{j,l} F_l$ for all $j\in\{0,...,M\}$.
\end{lemma}

The following Theorem then follows from Lemma~\ref{lemma1}.
\begin{theorem}\label{theorem1}
Let $M,N\in\{0\}\cup\mathbb{N}$, and let $\{E_j\}_{j=0}^M$ and $\{F_l\}_{l=0}^N$ be sets of non-zero Kraus operators for the same CPTP mapping $\mathcal{E}$. If $\{F_l\}_{l=0}^N$ is a linearly independent set of operators, then $M\geq N$.
\end{theorem}
\begin{proof}
Suppose $M<N$. Adding $N-M$ zero operators, $E_{M+1}=...=E_N=0$, to the set $\{E_j\}_{j=0}^M$, we have a new set of Kraus operators $\{E_j\}_{j=0}^N$.

Then by Lemma~\ref{lemma1}, $0=E_N=\sum\limits_{l=0}^N u_{N,l} F_l$, where $u_{j,l}$'s are elements of an unitary matrix $U$.

Now, since the $F_l$'s are linearly independent, $u_{N,l}=0$ for all $l$. This implies that a row of $U$ has elements all equal to $0$. Hence, $U$ does not have full rank and therefore cannot be unitary.

The contradiction implies that $M\geq N$.
\end{proof}

It is easy to check that the set of Kraus operators for the weak CPTP mapping given by Eq.~(\ref{eq:IIIA03}) is linearly independent. Therefore, by Theorem~\ref{theorem1},
\beq\label{eq:IIIA05}
|S|\geq 2^{n-k}+1.
\eeq
Hence, each of the Kraus operators $K_l$ may be written
\begin{align}\label{eq:IIIA06}
K_l &=\sum\limits_{j=0}^{2^{n-k}-1} u_{l,j}~\varepsilon R_j+u_{l,2^{n-k}}\sqrt{1-\varepsilon^2}~I^{\otimes n} \\
 &=u_{l,2^{n-k}}\sqrt{1-\varepsilon^2}~I^{\otimes n}+\sum\limits_{j=0}^{2^{n-k}-1} u_{l,j}~\varepsilon~ I^{\otimes k}\otimes\ket{\underline{0}}\bra{\underline{j}}
\end{align}
where $\{u_{l,k}\}_{l,k\in S}$ are elements of an unitary matrix.

The next step requires the polar decomposition of each $K_l$. For all $l\in S$, $K_l$ can be decomposed as 
\beq\label{eq:IIIA07}
K_l = U_{C,l} M_l,
\eeq
where $U_{C,l}$ is unitary and $M_l$ is positive-semidefinite. The operators $\{M_l\}_{l\in S}$ describe a positive-operator valued measurement (POVM). Each of the $M_l^\dag M_l$ is positive-semidefinite, and the completeness relation is satisfied:
\begin{align}
\sum\limits_{l\in S} M_l^\dag M_l &= \sum\limits_{l\in S} \left(U_{C,l}^\dag K_l\right)^\dag \left(U_{C,l}^\dag K_l\right)\label{eq:IIIA09} \\
&= \sum\limits_{l\in S} K_l^\dag K_l\label{eq:IIIA10} = I^{\otimes n}.
\end{align}

To implement the weak correction map given by Eq.~(\ref{eq:IIIA01}), we recall the four steps of CTQEC that were introduced in Sec.~\ref{IIB}:

(1) Couple the primary system to an ancillary system.

(2) Apply a weak unitary to both systems.

(3) Weakly measure and then discard the ancillary system.

(4) Apply a weak unitary correction, conditioned on the measurement result, to the system.

The polar decomposition gives a set of weak measurement operators $\{M_l\}_{l\in S}$ and a set of weak unitary operators $\{U_{C,l}\}_{l\in S}$. The mapping defined by the Kraus operators can be treated as first measuring with $\{M_l\}_{l\in S}$ and then applying a unitary correction conditioned on the measurement outcome. Therefore, the corrections in step 4 are given by the set $\{U_{C,l}\}_{l\in S}$, and steps 1 to 3 implement the measurement given by the set of measurement operators $\{M_l\}_{l\in S}$. The remaining question is how to carry out the three steps for the measurement. We shall answer this question in the rest of this section.

\subsection{Coherent measurement and the ancillary system}
\label{IIIB}

We begin with the required ancillary system. In step 3, ancillary system is measured. Since each outcome of the measurement on the ancillary system should result in a unique measurement operator from $\{M_l\}_{l\in S}$ acting on the primary system, there must exist a set of $|S|$ mutually orthogonal states in the ancillary system. So the dimension of the Hilbert space of the ancillary system must be no less than $|S|$. If the ancillary system has $s_A\in\mathbb{N}$ qubits then
\beq\label{eq:IIIB01}
2^{s_A}\geq |S|~\Rightarrow~s_A\geq log_2|S|.
\eeq

Therefore, by Eq.~(\ref{eq:IIIA05}), the minimum size of the ancilla, $\bar{s}_A$, is
\beq\label{eq:IIIB02}
\bar{s}_A=\lceil log_2\left( 2^{n-k}+1 \right)\rceil=n-k+1.
\eeq

Since we are investigating the general setup, we do not require the ancillary system to be of minimum size. Let $s_A\geq log_2|S|$ be the size of the ancillary system, and let the set $S_A=\{0,1,...,2^{s_A}-1\}$. We denote a set of standard basis states for this system as
\beq\label{eq:IIIB03}
\left\{\ket{\underline{l}_a}\right\}_{l\in S_A},
\eeq
where the underline notation is defined as in Sec.~\ref{II}, and the subscript ``$a$" indicates that the basis state is in the ancillary Hilbert space of dimension $2^{s_A}$. Throughout this paper, we will label the Hilbert space of the primary system $\mathcal{H}_S$ and the Hilbert space of the ancilla $\mathcal{H}_A$.

The initial state of the ancilla, $\ket{A_0}\in \mathcal{H}_A$, can be chosen such that each of the standard basis states $\{\ket{\underline{j}_a}\}_{j\in S}$ is equally weighted:
\beq\label{eq:IIIB04}
\ket{A_0}\equiv\frac{1}{\sqrt{|S|}}\sum\limits_{j\in S} \ket{\underline{j}_a}.
\eeq
Note that $S\subset S_A$, so $\ket{A_0}\neq\ket{+}^{\otimes s_A}$ in general.

Let $U_M$ be the desired weak unitary operator on the combined Hilbert space of both the primary and ancillary systems, $\mathcal{H}_S\otimes\mathcal{H}_A$. Suppose the primary system is initially in an arbitrary state $\ket{\psi}\in\mathcal{H}_S$. Then Eq.~(\ref{eq:IIIB04}) implies that the initial state for the combined system is $\ket{\psi}\otimes\ket{A_0}$. We choose $U_M$ to satisfy
\beq\label{eq:IIIB05}
U_M\left(\ket{\psi}\otimes\ket{A_0}\right)=\sum\limits_{j\in S}M_j\ket{\psi}\otimes\ket{\underline{j}_a}.
\eeq

After measuring the ancillary system in the basis $\{\ket{\underline{j}_a}\}_{j\in S_A}$, the state of the primary system after measurement result $j\in S$ is $M_j\ket{\psi}$, and any other outcomes in $S_A\setminus S$ have zero probability.

To find the weak unitary $U_M$, let
\beq\label{eq:IIIB06}
U_M=e^{i\varepsilon H_M}=I^{\otimes n+s_A}+i\varepsilon H_M-\frac{1}{2}\varepsilon^2 H_M^2+O(\varepsilon^3),
\eeq
where $H_M$ is Hermitian. We can always write $H_M$ in the following form:
\beq\label{eq:IIIB07}
H_M=\sum\limits_{j,l\in S_A} H_{j,l}\otimes\ket{\underline{j}_a}\bra{\underline{l}_a},
\eeq
where $H_{j,l}$'s act on the primary system and satisfy $H_{j,l}=H_{l,j}^\dag$ since $H_M$ is Hermitian.

Up to second order in $\varepsilon$, Eq.~(\ref{eq:IIIB05}) implies that for all $j\in S$,
\beq\label{eq:IIIB08}
\frac{1}{\sqrt{|S|}}\left( I^{\otimes {n+s_A}}+i\varepsilon\sum\limits_{l\in S_A} H_{j,l}-\frac{1}{2}\varepsilon^2\sum\limits_{l,m\in S_A} H_{j,l}H_{l,m} \right)=M_j,
\eeq
and for all $j\in S\setminus S_A$,
\beq\label{eq:IIIB09}
\frac{1}{\sqrt{|S|}}\left( I^{\otimes {n+s_A}}+i\varepsilon\sum\limits_{l\in S_A} H_{j,l}-\frac{1}{2}\varepsilon^2\sum\limits_{l,m\in S_A} H_{j,l}H_{l,m} \right)=0.
\eeq

Solving the above equations may be difficult in general. However, we show in the next section a special case where those equations can be solved.

\subsection{A protocol for quantum-jump CTQEC requiring minimal number of ancillas}
\label{IIIC}

In this section, we propose a specific CTQEC protocol based on the formalism in the previous sections. This protocol requires a minimal ancillary system of size $s_A=\bar{s}_A=n-k+1$, and the weak operators involved are in a form such that the weak Hermitian operator associated with the coherent measurement can be easily found.

The target error-correcting map is given by Eq.~(\ref{eq:III03}). Let $S=\{0,1,...,2^{n-k+1}-1\}$, and consider the following set of Kraus operators $\left\{ K_j\right\}_{j\in S}$ for this mapping:
\begin{multline}\label{eq:IIIC01}
K_j =\frac{1}{\sqrt{2^{n-k+1}}}I^{\otimes k}\otimes\Big(\sqrt{1-\varepsilon^2}~I^{\otimes n-k} \\
+i\varepsilon\sqrt{2^{n-k}}\ket{\underline{0}}\bra{\underline{j}}\Big),
\end{multline}
\begin{multline}\label{eq:IIIC02}
K_{2^{n-k}+j} =\frac{1}{\sqrt{2^{n-k+1}}}I^{\otimes k}\otimes\Big(\sqrt{1-\varepsilon^2}~I^{\otimes n-k} \\
-i\varepsilon\sqrt{2^{n-k}}\ket{\underline{0}}\bra{\underline{j}}\Big),
\end{multline}
for all $j\in\{0,1,...,2^{n-k}-1\}$. Note that in this case $S_A=S$.

Following our formalism, we first find the polar decomposition of these Kraus operators. This results in a unique set of POVM operators $\{ M_j\}_{j\in S}$ and a corresponding set of correcting unitaries $\{ U_{C,j}\}_{j\in S}$ satisfying $K_j=U_{C,j} M_j$ for all $j\in S$. Up to second order in $\varepsilon$, these operators are listed below:
\begin{multline}\label{eq:IIIC03}
M_0=\frac{1}{\sqrt{2^{n-k+1}}}~I^{\otimes k}\otimes\bigg[
\left(1-\frac{1}{2}\varepsilon^2\right) I^{\otimes n-k} \\
+2^{n-k-1} \varepsilon^2\ket{\underline{0}}\bra{\underline{0}}
\bigg],
\end{multline}
\begin{multline}\label{eq:IIIC04}
U_{C,0}=I^{\otimes k}\otimes\bigg[
I^{\otimes n-k}
+ \big(i \sqrt{2^{n-k}}\varepsilon \\
- 2^{n-k-1}\varepsilon^2\big) \ket{\underline{0}}\bra{\underline{0}}
\bigg],
\end{multline}
\begin{multline}\label{eq:IIIC05}
M_{2^{n-k}}=\frac{1}{\sqrt{2^{n-k+1}}} I^{\otimes k}\otimes\bigg[
\left(1-\frac{1}{2}\varepsilon^2\right)I^{\otimes n-k} \\
+ 2^{n-k-1} \varepsilon^2\ket{\underline{0}}\bra{\underline{0}}
\bigg],
\end{multline}
\begin{multline}\label{eq:IIIC06}
U_{C,2^{n-k}}=I^{\otimes k}\otimes\bigg[
I^{\otimes n-k}
+ \big(-i \sqrt{2^{n-k}}\varepsilon \\
- 2^{n-k-1}\varepsilon^2\big) \ket{\underline{0}}\bra{\underline{0}}
\bigg],
\end{multline}
and $\forall~j\in\{1,2,...,2^{n-k}-1\}$,
\begin{multline}\label{eq:IIIC07}
M_j=\frac{1}{\sqrt{2^{n-k+1}}} I^{\otimes k}\otimes\bigg[
\Bigg(1-\frac{1}{2}\varepsilon^2\Bigg) I^{\otimes n-k}\\
- 2^{n-k-3}\varepsilon^2\ket{\underline{0}}\bra{\underline{0}}
+ 3\cdot 2^{n-k-3}\varepsilon^2\ket{\underline{j}}\bra{\underline{j}}\\
+ i\sqrt{2^{n-k-2}}\varepsilon\left(\ket{\underline{0}}\bra{\underline{j}}-\ket{\underline{j}}\bra{\underline{0}}\right)
\bigg],
\end{multline}
\begin{multline}\label{eq:IIIC08}
U_{C,j}=I^{\otimes k}\otimes\bigg[
I^{\otimes n-k}
- 2^{n-k-3}\varepsilon^2\left(\ket{\underline{0}}\bra{\underline{0}}+\ket{\underline{j}}\bra{\underline{j}}\right) \\
+ i \sqrt{2^{n-k-2}}\varepsilon\left(\ket{\underline{0}}\bra{\underline{j}}+\ket{\underline{j}}\bra{\underline{0}}\right)
\bigg],
\end{multline}
\begin{multline}\label{eq:IIIC09}
M_{2^{n-k}+j}=\frac{1}{\sqrt{2^{n-k+1}}} I^{\otimes k}\otimes\bigg[
\Bigg(1-\frac{1}{2}\varepsilon^2\Bigg) I^{\otimes n-k}\\
- 2^{n-k-3}\varepsilon^2\ket{\underline{0}}\bra{\underline{0}}
+ 3\cdot 2^{n-k-3}\varepsilon^2\ket{\underline{j}}\bra{\underline{j}}\\
- i\sqrt{2^{n-k-2}}\varepsilon\left(\ket{\underline{0}}\bra{\underline{j}}-\ket{\underline{j}}\bra{\underline{0}}\right)
\bigg],
\end{multline}
\begin{multline}\label{eq:IIIC10}
U_{C,2^{n-k}+j}=I^{\otimes k}\otimes\bigg[
I^{\otimes n-k}
- 2^{n-k-3}\varepsilon^2\left(\ket{\underline{0}}\bra{\underline{0}}+\ket{\underline{j}}\bra{\underline{j}}\right)\\
- i \sqrt{2^{n-k-2}}\varepsilon\left(\ket{\underline{0}}\bra{\underline{j}}+\ket{\underline{j}}\bra{\underline{0}}\right)
\bigg],
\end{multline}

Note that the correcting unitaries may be written in the form $U_{C,j}=e^{i\varepsilon H_{C,j}}$, where $\forall j\in\{0,1,...,2^{n-k}-1\}$,

\begin{align}
H_{C,j} &=I^{\otimes k}\otimes\frac{\sqrt{2^{n-k}}}{2}\left(\ket{\underline{0}}\bra{\underline{j}}+\ket{\underline{j}}\bra{\underline{0}}\right),\label{eq:IIIC11} \\
H_{C,2^{n-k}+j} &= -I^{\otimes k}\otimes\frac{\sqrt{2^{n-k}}}{2}\left(\ket{\underline{0}}\bra{\underline{j}}+\ket{\underline{j}}\bra{\underline{0}}\right),\label{eq:IIIC12}
\end{align}

Now, as we have claimed, the ancillary system has size
\beq\label{eq:IIIC13}
s_A=n-k+1.
\eeq

In this case, $S_A=\{0,1,...,2^{n-k+1}-1\}=S$. Following Eq.~(\ref{eq:IIIB04}), the ancillary system is prepared in the state
\beq\label{eq:IIIC14}
\ket{A_0}=\frac{1}{\sqrt{2^{n-k+1}}}\sum\limits_{j\in S} \ket{\underline{j}_a}=\ket{+}^{\otimes n-k+1}.
\eeq

Finding $U_M$ is equivalent to finding the corresponding Hermitian operator $H_M$. We have mentioned that one may regard this task as solving Eqs.~(\ref{eq:IIIB08}) and~(\ref{eq:IIIB09}), and in fact, solving only Eq.~(\ref{eq:IIIB08}) is required since $S=S_A$. Since the $M_j$'s are known, by collecting the terms of each order in $\varepsilon$, then solving the following Eqs.~(\ref{eq:IIIC15})--(\ref{eq:IIIC19}) is equivalent to solving Eq.~(\ref{eq:IIIB08})
\beq\label{eq:IIIC15}
\sum\limits_{j\in S} H_{0,j}=\sum\limits_{j\in S} H_{2^{n-k},j}=0,
\eeq
\beq\label{eq:IIIC16}
\begin{aligned}
&\sum\limits_{j,l\in S} H_{0,j}H_{j,l}=\sum\limits_{j,l\in S} H_{2^{n-k},j}H_{j,l} \\
= &I^{\otimes k}\otimes\left(I^{\otimes n-k}-2^{n-k}\ket{\underline{0}}\bra{\underline{0}}\right),
\end{aligned}
\eeq
and $\forall~j\in\{1,2,...,2^{n-k}-1\}$,
\beq\label{eq:IIIC17}
\begin{aligned}
&\sum\limits_{l\in S} H_{j,l}=-\sum\limits_{l\in S} H_{2^{n-k}+j,l} \\
= &I^{\otimes k}\otimes\frac{\sqrt{2^{n-k}}}{2}\left(\ket{\underline{0}}\bra{\underline{j}}-\ket{\underline{j}}\bra{\underline{0}}\right),
\end{aligned}
\eeq
\beq\label{eq:IIIC18}
\begin{aligned}
&\sum\limits_{l,m\in S} H_{j,l}H_{l,m}=\sum\limits_{l,m\in S} H_{2^{n-k}+j,l}H_{l,m}, \\
= &I^{\otimes k}\otimes\bigg(I^{\otimes n-k}+\frac{2^{n-k}}{4}\ket{\underline{0}}\bra{\underline{0}}-\frac{3\cdot 2^{n-k}}{4}\ket{\underline{j}}\bra{\underline{j}}\bigg),
\end{aligned}
\eeq
and $\forall~j,l\in S$,
\beq\label{eq:IIIC19}
H_{j,l}=H_{l,j}^\dag.
\eeq

These equations potentially have many solutions. We solved the equations for a small problem size, and then generalized our solution to arbitrary problem size. The small problem we consider is three-qubit bit-flip code, where $n=3$, $k=1$, and therefore $s_A=n-k+1=3$. In this case, one can solve for the $\left\{ H_{j,l} \right\}$ straightforwardly by hand. Generalizing to arbitrary size involves assigning parameters as elements of the $H_{j,l}$'s similar to the solution to the three-qubit code case, and then solving for the parameters by plugging everything back into Eq.~(\ref{eq:IIIB08}). Here is this solution for the general case:
\beq\label{eq:IIIC20}
\begin{aligned}
H_{0,0} = &-H_{2^{n-k},2^{n-k}} \\
= &\frac{-2}{\sqrt{2^{n-k}}} I^{\otimes k}\otimes\sum\limits_{l=1}^{2^{n-k}-1} \left(\ket{\underline{0}}\bra{\underline{l}}+\ket{\underline{l}}\bra{\underline{0}}\right),
\end{aligned}
\eeq
\beq\label{eq:IIIC21}
H_{0,2^{n-k}} = H_{2^{n-k},0} = 0,
\eeq
and $\forall~j\in\{1,2,...,2^{n-k}-1\}$,
\beq\label{eq:IIIC22}
\begin{aligned}
H_{0,j} = &H_{j,0} \\
= &\frac{2}{\sqrt{2^{n-k}}} I^{\otimes k}\otimes \left(\ket{\underline{0}}\bra{\underline{j}}+\ket{\underline{j}}\bra{\underline{0}}\right),
\end{aligned}
\eeq
\beq\label{eq:IIIC23}
H_{0,2^{n-k}+j} = H_{2^{n-k}+j,0} =0,
\eeq
\beq\label{eq:IIIC24}
\begin{aligned}
H_{j,j} = &-H_{2^{n-k}+j,2^{n-k}+j} \\
= &\frac{2\cdot(1-2^{n-k})}{\sqrt{2^{n-k}}} I^{\otimes k}\otimes \left(\ket{\underline{0}}\bra{\underline{j}}+\ket{\underline{j}}\bra{\underline{0}}\right),
\end{aligned}
\eeq
\beq\label{eq:IIIC25}
H_{2^{n-k},j} = H_{j,2^{n-k}} =0,
\eeq
\beq\label{eq:IIIC26}
\begin{aligned}
H_{2^{n-k},2^{n-k}+j} = &H_{2^{n-k}+j,2^{n-k}} \\
= &\frac{-2}{\sqrt{2^{n-k}}} I^{\otimes k}\otimes \left(\ket{\underline{0}}\bra{\underline{j}}+\ket{\underline{j}}\bra{\underline{0}}\right),
\end{aligned}
\eeq
\beq\label{eq:IIIC27}
\begin{aligned}
H_{j,2^{n-k}+j} = &-H_{2^{n-k}+j,j} \\
&=\frac{\sqrt{2^{n-k}}}{2} I^{\otimes k}\otimes \left(\ket{\underline{0}}\bra{\underline{j}}-\ket{\underline{j}}\bra{\underline{0}}\right),
\end{aligned}
\eeq
and $\forall~j,l\in\{1,2,...,2^{n-k}-1\}$ where $j\neq l$,
\beq\label{eq:IIIC28}
\begin{aligned}
&H_{j,l} = -H_{2^{n-k}+j,2^{n-k}+l} \\
= &\frac{1}{\sqrt{2^{n-k}}} I^{\otimes k}\otimes \left(\ket{\underline{0}}\bra{\underline{j}}+\ket{\underline{j}}\bra{\underline{0}}+\ket{\underline{0}}\bra{\underline{l}}+\ket{\underline{l}}\bra{\underline{0}}\right),
\end{aligned}
\eeq
\beq\label{eq:IIIC29}
\begin{aligned}
&H_{j,2^{n-k}+l} = -H_{2^{n-k}+j,l} \\
= &\frac{1}{\sqrt{2^{n-k}}} I^{\otimes k}\otimes \left(\ket{\underline{0}}\bra{\underline{j}}+\ket{\underline{j}}\bra{\underline{0}}-\ket{\underline{0}}\bra{\underline{l}}-\ket{\underline{l}}\bra{\underline{0}}\right).
\end{aligned}
\eeq

The desired weak unitary is therefore $U_M=e^{i\varepsilon H_M}$, where
\beq\label{eq:IIIC30}
H_M=\sum\limits_{j,l\in S} H_{j,l}\otimes\ket{\underline{j}_a}\bra{\underline{l}_a}.
\eeq

\begin{exmp}\label{exmp2}
In this example, we demonstrate the protocol for the three-qubit bit-flip code.

The primary system is first coupled to an ancillary system of $s_A=n-k+1=3$ qubits:
\beq\label{eq:IIIC31}
\ket{A_0}=\ket{+}^{\otimes 3}.
\eeq

Then the unitary operation $U_M=e^{i\varepsilon H_M}$ is applied to the combined system, in which
\begin{align}
H_M &=\sum\limits_{j,l=0}^7 H_{j,l}\otimes\ket{\underline{j}_a}\bra{\underline{l}_a} \label{eq:IIIC32} \\
&\begin{multlined}
=\sum\limits_{j=1}^3 I\otimes\Big[ \ket{\underline{0}}\bra{\underline{j}}\otimes
\big(\ket{u_j}\bra{v_j}+\ket{v_j}\bra{w_j}\big) \\
\ket{\underline{j}}\bra{\underline{0}}\otimes
\left(\ket{v_j}\bra{u_j}+\ket{w_j}\bra{v_j}\right)\Big],
\end{multlined}\label{eq:IIIC33}
\end{align}
where
\beq\label{eq:IIIC34}
\ket{u_j}=\ket{+}\otimes\left(\ket{\underline{0}}-\sum\limits_{\substack{l=1\\ j<l}}^3 \ket{\underline{l}}\right),~~\ket{v_j}=\ket{-}\otimes\ket{\underline{j}},
\eeq
and
\beq\label{eq:IIIC35}
\ket{w_j}=\ket{+}\otimes\left(\ket{\underline{0}}+\sum\limits_{\substack{l=1\\ l\neq j}}^3 \ket{\underline{l}}-2\ket{\underline{j}}\right).
\eeq

Finally, the correcting unitary, $U_{C,j}$, corresponding to each measurement result $j\in\{0,1,2,...,7\}$ is given by Eqs.~(\ref{eq:IIIC04}),~(\ref{eq:IIIC06}),~(\ref{eq:IIIC08}) and~(\ref{eq:IIIC10}).
\end{exmp}

\section{Comparison with other CTQEC protocols}
\label{IV}

In this section, we will compare the proposed CTQEC protocol with two previous ones: those of Oreshkov~\cite{OreshkovChapter} and of Ahn, Doherty, and Landahl (ADL)~\cite{Ahn02}. In particular, for Oreshkov's method, we have investigated whether feedback strategy can improve its performance. However, we will show results that suggest that any improvement in performance may be minimal.

For a fair comparison, we set the error-correcting ``strength" of the protocols equal to each other. The measure of strength used is the diamond norm~\cite{Watrous11,Watrous09} of the difference between each weak map and the identity.

The contents related to Oreshkov's and ADL's schemes will be addressed in Sec.~\ref{IVA} and Sec.~\ref{IVB}, respectively.

\subsection{Oreshkov's protocol}
\label{IVA}

In~\cite{OreshkovChapter}, Oreshkov presents a method that does not perform the precise quantum-jump CTQEC map in Eq.~(\ref{eq:IIIA01}); however, the protocol gives the map when the primary system is initially in the code space and the error model has all correctable errors of the code occur as Poisson processes with the same rate $\lambda$. The procedure is very similar to that in this paper. One notable difference is that Oreshkov's protocol requires $2^{n-k}-1$ ancillary qubits; the result in the previous section shows that this number can be reduced (when $n-k>2$) to as low as $n-k+1$, which is the minimum number of ancillas possible! One might argue that each operator of the interaction Hamiltonian in Oreshkov's method is easier to realize, as it involves interaction between the primary system and a single ancilla qubit. However, each such operator is still a multi-body interaction, and the Hamiltonian includes exponentially many such terms.

Since Oreshkov's protocol and ours produce the same effective CTQEC mapping, the performances of both devices are the same.

One question about Oreshkov's protocol that we have investigated is whether feeding back all results from previous measurements would allow better correction. In general, the answer should be positive, but how much greater is the benefit? To answer this question, we first give a brief description of the protocol.

The procedure also follows the four steps (of coherent measurement and unitary correction) listed near the end of Sec.~\ref{IIIA}. In step 1, the ancillary system is prepared in the state
\beq\label{eq:IVA01}
\ket{\Psi_O}=\ket{+}^{\otimes 2^{n-k}-1}.
\eeq

In step 2, a joint weak unitary, $U_{M,O}=e^{-i\varepsilon H_{M,O}}$, is done, with Hamiltonian
\beq\label{eq:IVA02}
H_{M,O}=-\frac{1}{2} I^{\otimes k}\otimes\sum\limits_{j=1}^{2^{n-k}-1} X_{\underline{j}}\otimes Y_j^a,
\eeq
where $Y_j^a$ is the single Pauli-$Y$ operator acting on the $j$-th ancillary qubit, and
\beq\label{eq:IVA03}
X_{\underline{j}}=\ket{\underline{j}}\bra{\underline{0}}+\ket{\underline{0}}\bra{\underline{j}}.
\eeq

In step 3, a $Z$ measurement is performed on each of the ancillary qubits. If the measurement result on the $j$-th qubit is $m_j\in\{0,1\}$, then the unitary correction in step 4 is $U_{C,O}=e^{-i\varepsilon H_{C,O}}$ with
\beq\label{eq:IVA04}
H_{C,O}=\frac{1}{2}I^{\otimes k}\otimes\sum\limits_{j=1}^{2^{n-k}-1} (-1)^{m_j} Y_{\underline{j}_a},
\eeq
where
\beq
Y_{\underline{j}}=i\left(\ket{\underline{j}}\bra{\underline{0}}-\ket{\underline{0}}\bra{\underline{j}}\right).
\eeq

In this framework, the correcting operation is fixed. However, the parameter $\varepsilon$ in $U_{C,O}$ is allowed to vary with time, $U_{C,O}(t)=e^{-i\delta(t) H_{C,O}}$ for some function $\delta(t)$. Introducing this freedom can improve the performance, but we will see shortly that results for simple codes show minimal gain in performance, which suggests that a fixed correcting parameter already provides nearly optimal error-correcting capability.

Take the case of the three-qubit bit-flip code for example. Let $\rho_I$ be the density matrix of the information system that is to be protected from error. The initial state of the primary system in the corrected basis (and also the encoded basis) is then
\beq\label{eq:IVA05}
\rho(0)=\rho_I\otimes\ket{\underline{0}}\bra{\underline{0}}.
\eeq

Suppose the error model is independent bit-flip Poisson processes on each qubit with rate $\lambda$. Then the density matrix of the primary system at time $t$ is
\beq\label{eq:IVA06}
\rho(t)=\sum\limits_{j=0}^3 w_j(t)\rho_j,
\eeq
where the $\rho_j$'s are normalized density matrices corresponding to $j$ different errors applied to the initial code state:
\begin{align}
\rho_0 &=\rho(0),\label{eq:IVA07} \\
\rho_1 &=\frac{1}{3}\sum\limits_{j=1}^3 E_{j(C)}\rho(0)E_{j(C)},\label{eq:IVA08} \\
\rho_2 &=\frac{1}{3}\sum\limits_{\substack{j,l=1\\ j<l}}^3 E_{j(C)}E_{l(C)}\rho(0)E_{j(C)}E_{l(C)},\label{eq:IVA09} \\
\rho_3 &=E_{1(C)}E_{2(C)}E_{3(C)}\rho(0)E_{1(C)}E_{2(C)}E_{3(C)},\label{eq:IVA10}
\end{align}
where the $E_{j(C)}$'s are the bit-flip errors represented in the corrected basis (as defined in Example~\ref{exmp1}). The $w_j(t)$'s are real functions that form a set of probabilities at each time $t$. Note that $w_0(0)=1$.

Suppose at time $t$ the density matrix of the primary system is given by Eq.~(\ref{eq:IVA06}). Then by applying the CTQEC map in Eq.~(\ref{eq:III03}) to this density matrix,
\begin{align}
w_0(t) &\rightarrow w_0(t)\left(1-\frac{3(\varepsilon-\delta(t))^2}{4}\right)+w_1(t)\frac{(\varepsilon+\delta(t))^2}{4},\label{eq:IVA15}\\
w_1(t) &\rightarrow w_1(t)\left(1-\frac{(\varepsilon+\delta(t))^2}{4}\right)+w_0(t)\frac{3(\varepsilon-\delta(t))^2}{4},\label{eq:IVA16}\\
w_2(t) &\rightarrow w_2(t)\left(1-\frac{(\varepsilon+\delta(t))^2}{4}\right)+w_3(t)\frac{3(\varepsilon-\delta(t))^2}{4}\label{eq:IVA17}\\
w_3(t) &\rightarrow w_3(t)\left(1-\frac{3(\varepsilon-\delta(t))^2}{4}\right)+w_2(t)\frac{(\varepsilon+\delta(t))^2}{4}.\label{eq:IVA18}
\end{align}

Since the goal of error correction is to map the state back to the $\ket{\underline{0}}$-subspace, the parameter in this CTQEC map is chosen to maximize the updated $w_0(t)$. Hence, the optimal strategy would be to choose $\delta(t)$ as follows:
\beq\label{eq:IVA19}
\begin{aligned}
\delta(t) &=\argmax_\delta\left\{w_0(t)\left(1-\frac{3(\varepsilon-\delta)^2}{4}\right)+w_1(t)\frac{(\varepsilon+\delta)^2}{4}\right\} \\
&=\frac{3 w_0(t)+w_1(t)}{3 w_0(t)-w_1(t)}\varepsilon.
\end{aligned}
\eeq

This is optimal correction. Now let $\varepsilon=\sqrt{\kappa\delta t}$, in order to realize the infinitesimal CTQEC map in Eq.~(\ref{eq:IIB01}) and let $\delta t\rightarrow 0$, the first-order differential equations for the $w_j(t)$'s including both the error process and error correction are
\beq\label{eq:IVA21}
\begin{aligned}
w_0^{\prime} &=-3\lambda w_0+\lambda w_1+\kappa\frac{3w_0 w_1}{3w_0-w_1}, \\
w_1^{\prime} &=3\lambda w_0-3\lambda w_1+2\lambda w_2-\kappa\frac{3w_0 w_1}{3w_0-w_1},\\
w_2^{\prime} &=3\lambda w_3-3\lambda w_2+2\lambda w_1-\kappa\frac{9w_0^2 w_2-3w_1^2 w_3}{(3w_0-w_1)^2},\\
w_3^{\prime} &=-3\lambda w_3+\lambda w_2+\kappa\frac{9w_0^2 w_2-3w_1^2 w_3}{(3w_0-w_1)^2}.
\end{aligned}
\eeq

The last terms in Eq.~(\ref{eq:IVA21}) are $\kappa w_1$, $-\kappa w_1$, $-\kappa w_2$ and $\kappa w_2$, respectively, if one takes $\delta(t)=\varepsilon$ to be a fixed value as Oreshkov considered.

Fig.~\ref{fig:IVA01} compares the performances of optimal and constant correction in Oreshkov's protocol in terms of codeword fidelity (i.e., $w_0(t)$) and correctable overlap (i.e., $w_0(t)+w_1(t)$, which is the codeword fidelity after a strong error correction is applied) between the cases of optimal and constant correction in Oreshkov's device. The rate of error-correction in this comparison is $\kappa=100\lambda$. We can see that the difference in performance is negligibly small. In fact, we can see from Figs.~\ref{fig:IVA01} and~\ref{fig:IVA02} that $w_1(t)$ is very small compared to $w_0(t)$ at any time $t$. This should not be surprising, since the correction terms in Eq.~(\ref{eq:IVA21}) drive the weights $w_1$ and $w_2$ towards $w_0$ and $w_3$, respectively, and the rate of this driving process is much larger than the rate of the error process ($\kappa\gg\lambda$). If we use the relationship $w_1\ll w_0$ to approximate the differential equations given by Eq.~(\ref{eq:IVA21}), we get
\begin{align}
w_0^{\prime} &\approx -3\lambda w_0+\lambda w_1+\kappa w_1+\kappa\frac{w_1^2}{3w_0},\label{eq:IVA25}\\
w_1^{\prime} &\approx 3\lambda w_0-3\lambda w_1+2\lambda w_2-\kappa w_1-\kappa\frac{w_1^2}{3w_0},\label{eq:IVA26}\\
w_2^{\prime} &\approx 3\lambda w_3-3\lambda w_2+2\lambda w_1-\kappa w_2-\kappa\frac{w_1(2w_0 w_2-w_1 w_3)}{(3w_0^2)^2},\label{eq:IVA27}\\
w_3^{\prime} &\approx -3\lambda w_3+\lambda w_2+\kappa w_2+\kappa\frac{w_1(2w_0 w_2-w_1 w_3)}{(3w_0^2)^2}.\label{eq:IVA28}
\end{align}

The last terms in the above approximation are the result of optimal correction. Since $w_1\ll w_0$, the correction is dominated by the second to last terms, and the effect of the extra terms is negligibly small. The extra effort to track the state and apply precise corrections yields little benefit.

\begin{figure}
\includegraphics[width= 3.5in]{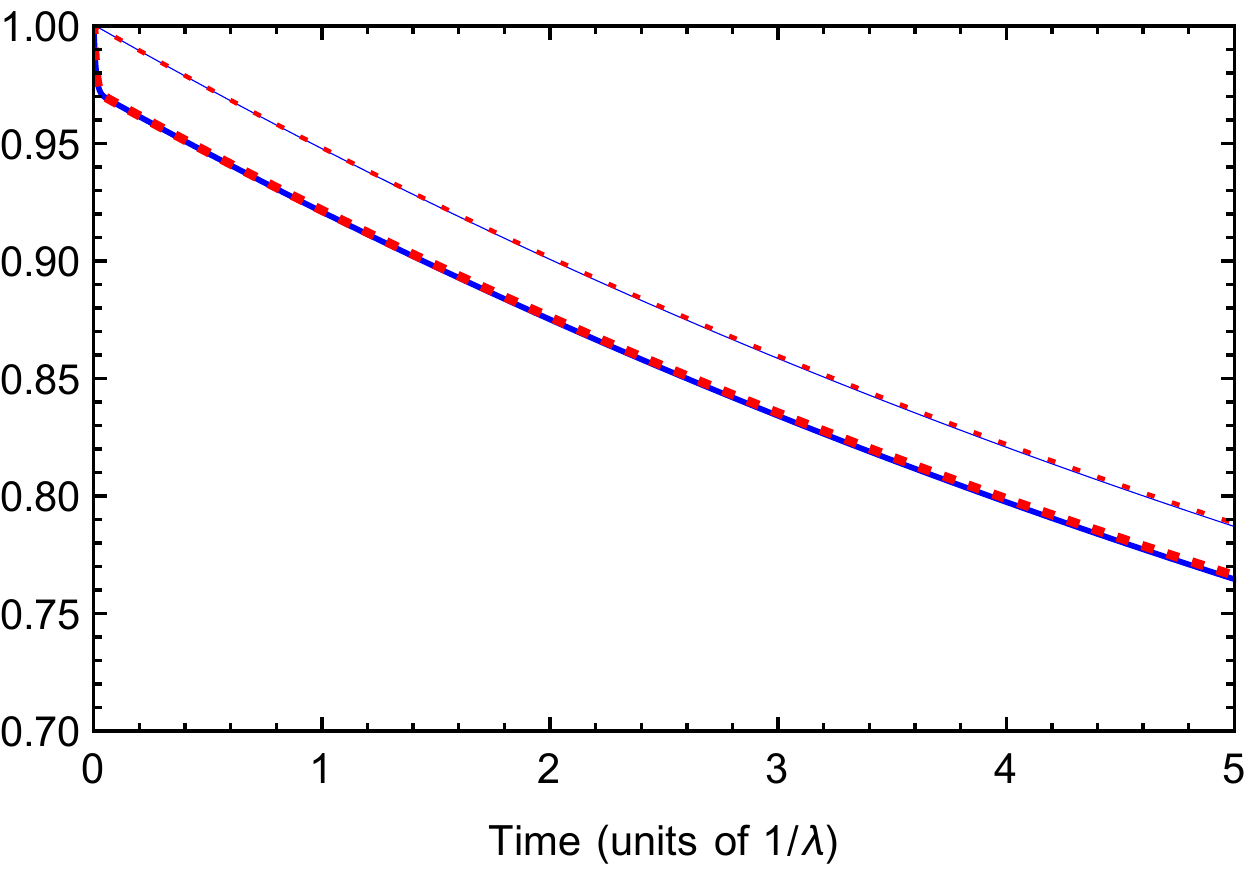}
\caption{\label{fig:IVA01}(Color online) Comparison of CTQEC performance between optimal and constant correction for Oreshkov's device using the three-qubit bit-flip code under independent bit-flip channels with error rate $\lambda$ and correcting parameter $\kappa=100\lambda$. The thin and thick dashed red lines are, respectively, the correctable overlap and codeword fidelity of the case with optimal correction. The thin and thick solid blue lines are, respectively, the correctable overlap and codeword fidelity of the case with constant correction.}
\end{figure}

\begin{figure}
\includegraphics[width= 3.5in]{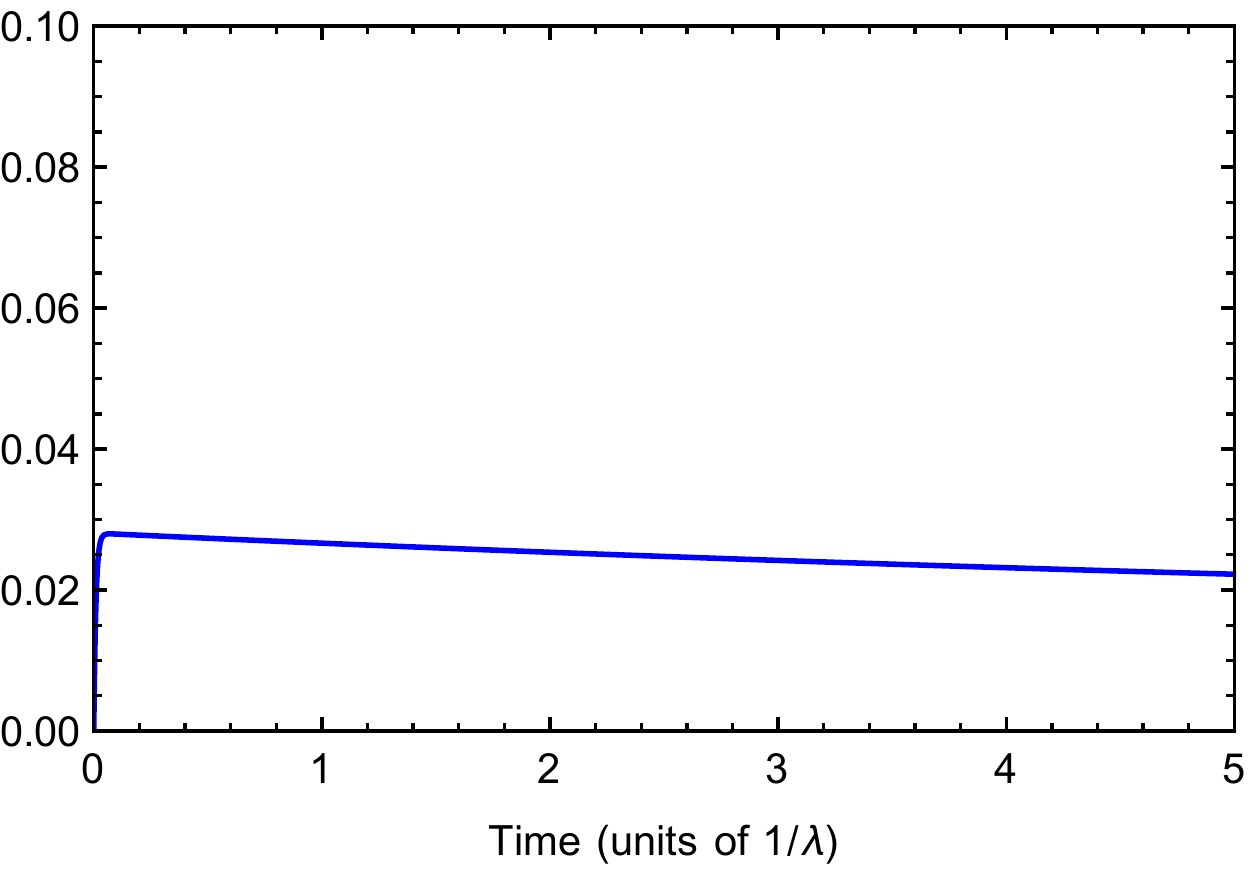}
\caption{\label{fig:IVA02}(Color online) Plot of $w_1(t)$ in the case of optimal correction for Oreshkov's device with $\kappa=100\lambda$.}
\end{figure}

In Fig.~\ref{fig:IVA03}, we show a more complicated example using the five-qubit ``perfect" code that can correct any single-qubit error. The correcting parameter is again chosen to be $\kappa=100\lambda$. We see that the benefit of optimal correction is still very small.

Note that, in practice, $\kappa$ will probably be larger than $100\lambda$, and the benefit of optimal correction will be even smaller. Therefore, we conclude that the constant correction method considered in Oreshkov's protocol has good enough error-correcting performance.

\begin{figure}
\includegraphics[width= 3.5in]{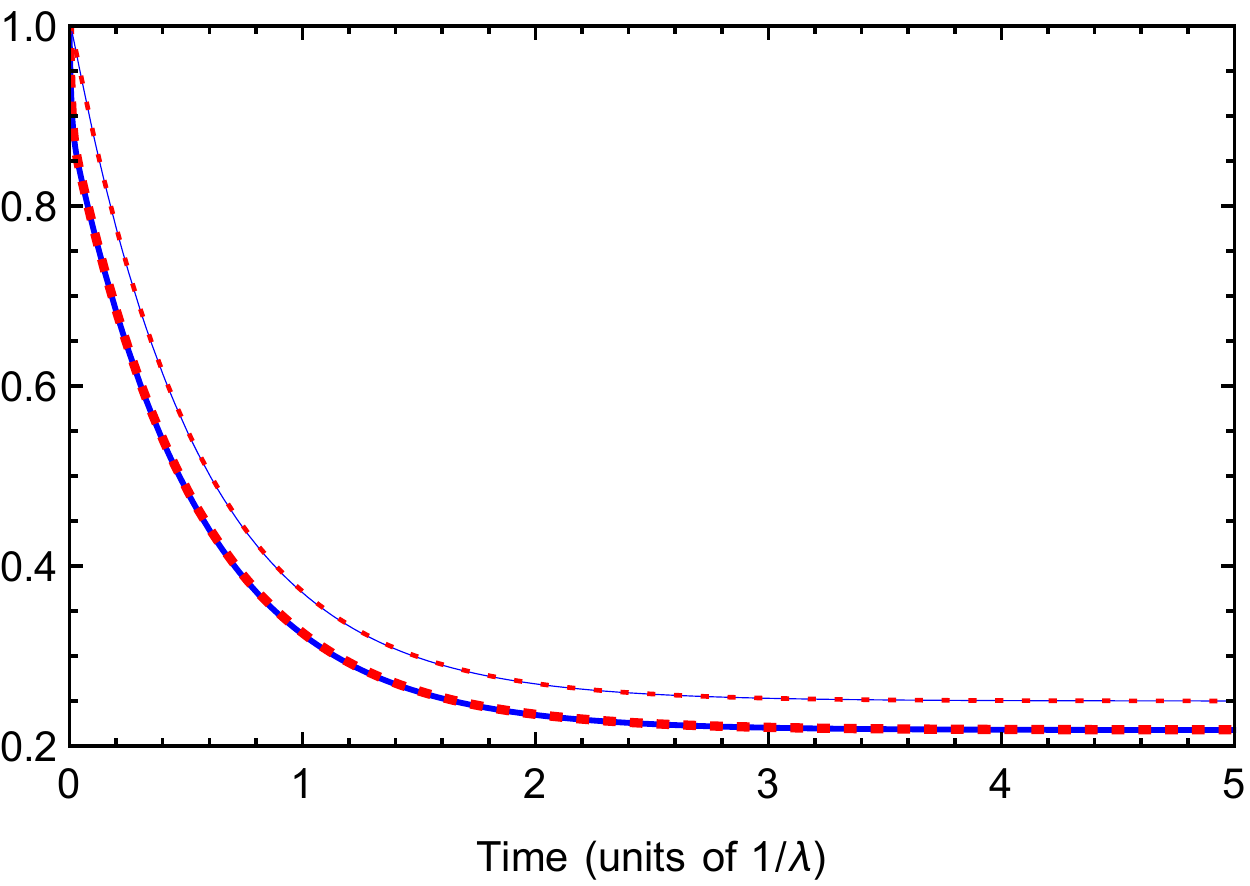}
\caption{\label{fig:IVA03}(Color online) Comparison of CTQEC performance between optimal and constant correction for Oreshkov's protocol using the five-qubit ``perfect" code under independent depolarizing channels with error rate $\lambda$ and correcting parameter is $\kappa=100\lambda$. The thin and thick dashed red lines are, respectively, the correctable overlap and codeword fidelity of the case with optimal correction. The thin and thick solid blue lines are, respectively, the correctable overlap and codeword fidelity of the case with constant correction.}
\end{figure}

\subsection{ADL protocol}
\label{IVB}

In~\cite{Ahn02}, Ahn, Doherty, and Landahl proposed a CTQEC protocol that weakly measures the stabilizer generators. The correcting Hamiltonian is a sum of correctable errors with tuning parameters that are controlled by the measurement results.

This protocol results in a stochastic master equation giving the quantum trajectories of the system. Note that the master equation of the system in our CTQEC protocol given by Eq.~(\ref{eq:IIB02}) does not have diffusive terms, which may arise in general in a quantum trajectory, since it is a map that averages over all measurement outcomes. The averaged map makes analysis and simulation easier, and it should accurately model the protocol since it represents the ``expected" evolution of the system. Therefore, it is reasonable to use the averaged map of the ADL protocol in the comparison.

To compare the performance between different protocols, we should adjust the parameters in both protocols so that the comparison is ``fair" in the sense that the strengths of the correcting maps are comparable. One commonly measures the strength of a quantum map $\Phi$ by the diamond norm $||\cdot||_\diamond$~\cite{Watrous11,Watrous09} of the difference between $\Phi$ and the identity map $\mathbb{I}$. Let this function be denoted by $\mathcal{D}(\cdot)$, where for all quantum maps $\Phi$,
\beq
\mathcal{D}(\Phi)\equiv||\Phi-\mathbb{I}||_\diamond.
\eeq

Now, we will compare our protocol to the ADL device for the three-qubit bit-flip code. Let $\Phi$ be the averaged CTQEC map for our protocol over an arbitrary small time $\delta t$. Then for any density matrix $\rho$,
\beq
\Phi(\rho)=(1-\kappa\delta t)\rho+\kappa\delta t \mathcal{R}(\rho),
\eeq
where we recall that $\mathcal{R}(\cdot)$ is the strong error-correcting map given by Eq.~(\ref{eq:III01}). Over the same time $\delta t$, given $\rho_c$ as the estimated density matrix based on previous measurement results, let $\Phi_{ADL,\rho_c}$ be the averaged CTQEC map for the ADL protocol that takes as input the density matrix $\rho$.
\beq
\begin{split}
\Phi_{ADL,\rho_c}(\rho)=\rho+\kappa_2\left(\tilde{\mathcal{L}}(g_1)+\tilde{\mathcal{L}}(g_2)+\tilde{\mathcal{L}}(g_1 g_2)\right)\rho\delta t\\
-\gamma_2 i \left[F\big(\eta_1(\rho_c),\eta_2(\rho_c),\eta_3(\rho_c)\big),\rho\right]\delta t,
\end{split}
\eeq
where $[\cdot,\cdot]$ is the commutator, $\tilde{\mathcal{L}}(\cdot)$ is a superoperator where for any operators $g$ and $\rho$, $\tilde{\mathcal{L}}(g)\rho$ is defined as
\beq
\tilde{\mathcal{L}}(g)\rho\equiv g\rho g^\dag-\frac{1}{2}g^\dag g\rho-\frac{1}{2}\rho g^\dag g,
\eeq
and the feedback Hamiltonian
\beq
F\left(\eta_1(\rho_c),\eta_2(\rho_c),\eta_3(\rho_c)\right)\equiv\sum\limits_{j=1}^3\eta_j(\rho_c) E_j,
\eeq
where $\eta_j(\cdot,\cdot)\in\{1,-1\}$. We do not specify the functions $\left\{\eta_j(\cdot,\cdot)\right\}$ as we will see that they are irrelevant in the case to be considered.

Note that the $g_j$'s and $E_j$'s are the stabilizer generators and correctable errors defined in Example~\ref{exmp1}.

We consider the ADL case given as an example in~\cite{Ahn02}. In this case, $\kappa_2=64\lambda$ and $\gamma_2=128\lambda$, where $\lambda$ is again the rate of the i.i.d. bit-flip error processes. If we define the mapping $\Phi_{ADL}^{(j,l,m)}$, for any operator $\rho$ and any $j,l,m\in\{1,-1\}$, as
\beq
\begin{split}
\Phi_{ADL}^{(j,l,m)}(\rho)\equiv\rho+\kappa_2\left(\tilde{\mathcal{L}}(g_1)+\tilde{\mathcal{L}}(g_2)+\tilde{\mathcal{L}}(g_1 g_2)\right)\rho\delta t\\
-\gamma_2 i\left[F\left(j,l,m\right),\rho\right]\delta t.
\end{split}
\eeq
Using semidefinite programming for computing diamond norm~\cite{Watrous09,cvx}, we find that for all $j,l,m\in\{1,-1\}$,
\beq
\mathcal{D}(\Phi_{ADL}^{(j,l,m)})=2\times 7.6847~\kappa_2\delta t.
\eeq

Hence, the map $\Phi_{ADL,\rho_c}$ would have
\beq
\mathcal{D}(\Phi_{ADL,\rho_c})=\mathcal{D}(\Phi_{ADL}^{(\eta_1(\rho_c),\eta_2(\rho_c),\eta_3(\rho_c))})=2\times 7.6847~\kappa_2\delta t.
\eeq

On the other hand, calculation shows
\beq
\mathcal{D}(\Phi)=2\kappa\delta t.
\eeq

Therefore, a fair comparison would be to choose $\kappa$ such that $\mathcal{D}(\Phi)=\mathcal{D}(\Phi_{ADL,\rho_c})$, which implies that $\kappa$ should be
\beq
\kappa=7.6847\kappa_2=7.6847\times 64\lambda.
\eeq

\begin{figure}[h!]
\includegraphics[width= 3.5in]{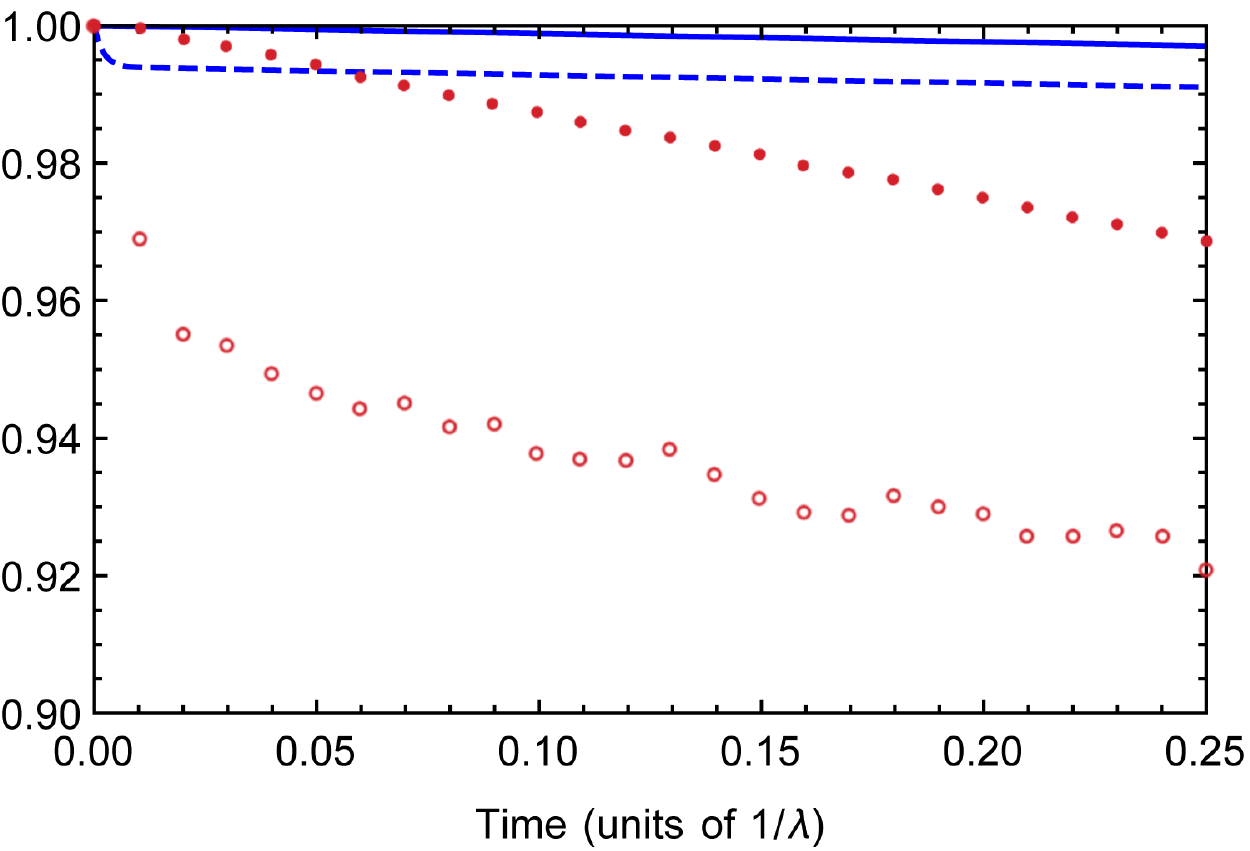}
\caption{\label{fig:IVB01}(Color online) The CTQEC performance of our protocol and ADL protocol using three-qubit bit-flip code under independent bit-flip channels with error rate $\lambda$. The parameters of the ADL protocol are $\kappa_2=64\lambda$ and $\gamma_2=128\lambda$, and the parameter of our protocol is $\kappa=7.6847\times 64\lambda$. The solid and dashed blue lines are, respectively, the correctable overlap and the codeword fidelity of our protocol. The filled and unfilled red circles are, respectively, the correctable overlap and the codeword fidelity of ADL's simulation of their protocol~\cite{Ahn02} sampled at time increments in $0.01$ $(1/ \lambda)$.}
\end{figure}

Given the initial state $\rho_0=\ket{\underline{0}}\bra{\underline{0}}=\ket{000}\bra{000}$, Fig.~\ref{fig:IVB01} shows the codeword fidelity and correctable overlap (i.e., the codeword fidelity after a strong error correction is applied) of our protocol when $\kappa=7.6847\times 64\lambda$ and ADL's simulation of their protocol~\cite{Ahn02}. We see that both the code fidelity and correctable overlap are much better using our protocol. The reason for this improvement is not entirely obvious. One reason we suspect for the improvement is that the intuitive approach in ADL -- weakly measuring the stabilizer generators -- is not optimal.

We can see this in the simple example of protecting a single qubit state $\ket{0}$ from bit-flip errors. The error-correcting code in this case has a single stabilizer generator $Z$. On the Bloch sphere, the bit-flip error process gradually brings the state from +1 towards −1 on the $z$-axis. The most obvious approach using strong error correction would measure the stabilizer $Z$ and then apply a correction $X$ if necessary. However, this is not the only approach: one could actually measure along any axis of the Bloch sphere, followed by a unitary correction conditioned on the measurement outcome. For strong error-correction, all of these protocols work equally well.

In the limit of weak measurements, however, the weak limit of these protocols are no longer equivalent. It can be shown that weak measurement along a direction perpendicular to the $z$-axis, followed by a weak unitary correction, will yield the greatest averaged increase in the purity of the state~\cite{Jacobs04}. This implies that for continuous-time quantum error correction, instead of weakly measuring the stabilizer $Z$ which is along the direction parallel to the $z$-axis, one should measure along a direction on the $x$-$y$ plane. For example, one could weakly measure $X$ or $Y$ to yield the most averaged increase in the purity of the state, followed by a weak rotation about the $Y$ or $X$ axis to maximize the fidelity with $\ket{0}$. Going from this simple example to a more general code suggests that weakly measuring the stabilizer generators will not increase the overlap with the code space as much as weakly measuring a complementary set of operators.

\section{Conclusion}
\label{V}

In this paper, we investigate and formalize the structure of quantum-jump CTQEC protocol involving weak measurements and weak unitary corrections. We show that for a given [[$n,k,d$]] quantum stabilizer code the minimum required size of the ancillary system is $n-k+1$ qubits, and we propose a method achieving this minimal requirement. We compare the performance of the proposed protocol to other known CTQEC protocols by Oreshkov and ADL. Our protocol is effectively equivalent to Oreshkov's quantum-jump CTQEC protocol, and hence the performance is the same. However, we reduce the number of ancillas required by an exponential factor. We also consider an improved protocol by optimizing the error-correcting parameter at each step, and we conclude after some analysis that this does not yield much benefit in general. For the comparison with the ADL protocol, we use the diamond norm to equalize the error-correcting strengths of the two maps and then compare their performances for the same code and noise model. Our protocol significantly outperforms the ADL protocol using the three-qubit bit-flip code against $X$ noise.

A critical issue that requires future investigation is that the Hamiltonian terms coupling the primary system to the ancillas are in general multi-body operators, which would be hard to implement experimentally. One possible approach uses the fact that our proposed formalism actually gives a family of quantum-jump CTQEC protocols. It would be interesting to see if there are members of the family where the terms of the Hamiltonian are low weight while still using only a modest number of ancillas.

\begin{acknowledgments}
K.-C.H. and T.A.B. thank an anonymous referee for useful suggestions. This research was supported by the ARO MURI under Grant No. W911NF-11-1-0268.
\end{acknowledgments}

\providecommand{\noopsort}[1]{}\providecommand{\singleletter}[1]{#1}%

\end{document}